\documentclass[prl,10pt,twocolumn,superscriptaddress,floatfix,showpacs]{revtex4-2}

\usepackage[utf8]{inputenc}  
\usepackage[T1]{fontenc}   
\usepackage[british]{babel} 
\usepackage[scaled=0.86]{berasans} 
\usepackage[colorlinks=true, allcolors=blue, urlcolor=blue]{hyperref} 
\usepackage{graphicx} 
\usepackage[babel]{microtype} 
\usepackage{amsmath,amssymb,amsthm,bm,amsfonts,mathrsfs,bbm} 

\usepackage{xspace} 
\usepackage{pgfplots}
\usepackage{xcolor,colortbl}
\usepackage{array}
\usepackage{bigstrut}
\usepackage{comment}

\newcommand{\bpm}{\begin{pmatrix}}
\newcommand{\epm}{\end{pmatrix}}

\newcommand{\ket}[1]{| #1 \rangle}
\newcommand{\bra}[1]{\langle #1|}

\newcommand{\ketbra}[1]{\ket{#1}\!\bra{#1}}
\newcommand{\be}{\begin{equation}}
\newcommand{\ee}{\end{equation}}
\newcommand{\bea}{\begin{eqnarray}}
\newcommand{\eea}{\end{eqnarray}}
\newcommand{\bes}{\begin{equation*}}
\newcommand{\ees}{\end{equation*}}
\newcommand{\bean}{\begin{eqnarray*}}
\newcommand{\eean}{\end{eqnarray*}}

\newtheorem*{thm*}{Theorem}

\newtheorem*{lem*}{Lemma}

\newtheorem*{lipschitzLem*}{Lemma \ref{lipschitz}}
\newtheorem*{lipschitzCubeLem*}{Lemma \ref{lipschitzCube}}
\newtheorem*{pgmNearlyOptimalThm*}{Theorem \ref{pgmNearlyOptimal}}

\newtheorem{theorem}{Theorem}
\newtheorem{proposition}[theorem]{Proposition}
\newtheorem{lemma}[theorem]{Lemma}

\theoremstyle{definition}

\newcommand{\no}[1]{\left\|#1\right\|} 
\newcommand{\tr}[1]{\mathrm{tr}\left[#1\right]} 
\newcommand{\id}{\mathbbm{1}} 

\newcommand{\parti}{\mathcal{P}} 
\newcommand{\M}{\mathsf{M}}
\newcommand{\N}{\mathsf{N}}

\newcommand{\ppre}{P_{\mathrm{pre}}}
\newcommand{\ppost}{P_{\mathrm{post}}^{\mathrm{fix}}}
\newcommand{\pposta}{P_{\mathrm{post}}}

\newcommand{\paux}{P_{\mathrm{aux}}}

\newcommand{\xaux}{X_{\mathrm{aux}}}

\newcommand{\madap}{\M^{\mathrm{post}}}

\begin{document}

\title{Later Information Can Outperform Earlier Information in Query-Assisted Quantum State Discrimination}

\author{Hanwool Lee }
\email{hanwool.h.lee@jyu.fi}

\author{Teiko Heinosaari}
\email{teiko.heinosaari@jyu.fi}
\affiliation{ Faculty of Information Technology, University of Jyväskylä, Finland}

\begin{abstract}
    Classical partial information about the states can increase the success probability in quantum state discrimination. An important aspect is its timing, whether it is given before or after the measurement. One might expect that later information is less useful than earlier information, because earlier information can still be stored and exploited at a later stage.
    Surprisingly, we show that the classical information can be more useful if it arrives after the measurement.
    We present a framework for state discrimination in which the sender provides classical information in response to a request from the receiver, which we refer to as a \textit{query}. 
    In some cases, nonorthogonal states can be perfectly discriminated when the query is allowed to depend on the measurement outcome.
    We give a general method for finding the optimal strategy in this setting, which turns out to be the standard minimum-error discrimination problem for an auxiliary ensemble. We show that for any number of states greater than two, there exists a qubit ensemble exhibiting an advantage of later information.
\end{abstract}

\maketitle

\textit{Introduction.} One of the fundamental facts in quantum information is that nonorthogonal states cannot be perfectly discriminated without some probability of error.
A state discrimination task can be eased by giving a possibility to combine quantum information with partial classical information \cite{BaWeWi08,GoWe10,CaHeTo18,carmeli2022quantum}.
For instance, if the task is to discriminate several quantum states, we can think of allowing one yes-no question before one has to guess the correct state.
A key characteristic of this type of question is its timing, whether it occurs before or after the measurement has been carried out. 

It is clear that a question, which we hereafter call a \textit{query}, asked before the measurement is useful and effective, as one can simply rule out half of the states and choose the measurement accordingly, i.e.,  to be the optimal measurement for the smaller, remaining subset of states.
A series of earlier investigations \cite{VaAhAl87,heinosaari2023anticipative, heinosaari2025metainformation,carmeli2025postponing} have shown that a query sent after the measurement has been performed can also be effective. In that case, the crucial point is that one decides the query before fixing the measurement. 
The optimal measurement is then called \emph{anticipative}, as it anticipates forthcoming partial classical information. This scenario naturally arises in quantum key distribution, such as the BB84 protocol, in which the encoding basis is revealed after the receiver has performed a measurement \cite{BeBr84,BaWeWi08}. This is also relevant to communication scenarios without quantum memory, in which one may have to perform a measurement before receiving classical information \cite{damgaard2008cryptography}. 

Two immediate observations can be made. First, post‑measurement information cannot outperform pre‑measurement information, since any classical information could simply be withheld until after the measurement. Second, neither approach enables perfect discrimination of a set of pure quantum states, if they are all pairwise nonorthogonal.
This follows from the fact that pre-measurement information leads to the usual kind of state discrimination, but of a smaller yet indistinguishable subset.

In the current work, we show that a small adaptation to the previously studied measurement strategies can have a significant effect. In our framework, the query sent after the measurement can depend on the obtained measurement outcome, referred to as \textit{adaptive} query. This seemingly slight change has, in fact, dramatic consequences.

The main result of this work is a general method for constructing the optimal strategy with post-measurement information. 
We show that the optimal measurement with an adaptive query is a minimum-error measurement for an auxiliary ensemble, which consists of all pairwise mixtures of the original ensemble.

This leads to several interesting observations.  First, we show that the number of states plays a crucial role in determining when posterior classical information can be advantageous. In some cases, nonorthogonal states can be perfectly discriminated only if the classical information arrives after the measurement. 
Furthermore, we show that the previously known dimensional bound \cite{PhysRevA.99.020102} on the number of perfectly distinguishable states can be exceeded with an adaptive query.

As quantum state discrimination is a necessary step in various information processing protocols \cite{BaCr09,Bergou10,BaKw15}, we expect that our method can be useful in hybrid classical-quantum protocols, providing a maximally effective way to combine classical and quantum information.
	
\begin{figure}
    \centering
    \includegraphics[width=0.9\linewidth]{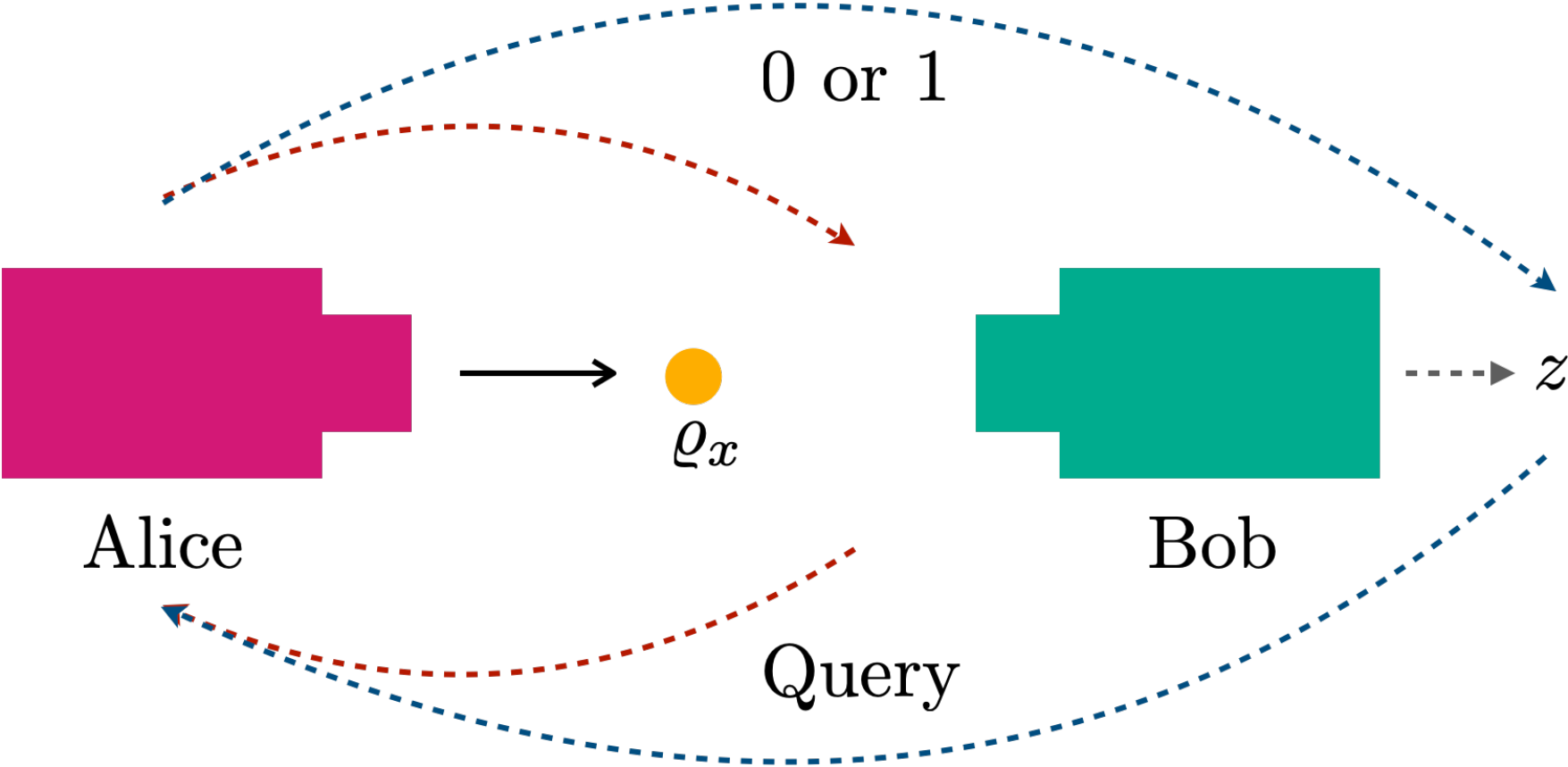}
    \caption{A scenario of query-assisted quantum state discrimination is illustrated. Bob sends a query to Alice, and Alice replies with a bit. The timing of the query and the classical information relative to the measurement significantly affects the optimal strategy and the achievable success probability.  Both the query and the bit can be exchanged either before (red) or after (blue) the measurement. 
        Interestingly, sending the query after the measurement can be more beneficial.  }
    \label{fig:placeholder}
\end{figure}

\textit{State discrimination with a query.}
In the standard quantum state discrimination task, a symbol $x \in X=\{1,\ldots n\}$ is drawn with prior probability $q_x$, and Alice sends the receiver, Bob, the corresponding state $\varrho_x$. 
In minimum-error discrimination, the receiver aims to minimize the average error probability, and the maximum \textit{success probability} is given by
\bea
P= \max_{\M}  \sum_{x=1}^n q_x \tr{\varrho_x \M_x}
\eea
where $\M$ is a measurement, mathematically described as a positive-operator valued measure (POVM), i.e. $\sum_x \M_x=\id$, and $\M_x \geq 0$ for $x \in X$. 

We consider a variant of this task, in which the receiver can send a query and gets a truthful answer to it.
The task can obviously become trivial if the amount of classical information from the sender to the receiver is not limited. For instance, if Alice can send $\lceil \log n \rceil$ bits of information, then she can simply tell the correct symbol to Bob.
For this reason, we focus on the case where Alice can only send one bit of classical information, so Bob must choose his question carefully to maximize the guessing probability. 
In this work, we focus on a single bit of classical information for simplicity, although the framework can be extended to any number of bits.

In our framework, the optimal discrimination strategy is characterized by three operations: a query, a measurement, and classical post-processing, whose relative timing significantly affects the guessing probability. Mathematically, a binary query corresponds to a partition, denoted by $\parti$, of the symbol set $X$ into two subsets,
\bea \label{eq:parti}
\parti: X=X_0 \cup X_1. 
\eea
Posing a query means that Bob chooses a partition and informs Alice of it; Alice then truthfully reports back which subset, $X_0$ or $X_1$, contains the true symbol by sending a bit.

As shown in the earlier works \cite{heinosaari2023anticipative,heinosaari2025metainformation}, the timing of classical partial information plays a critical role in discrimination strategy. Two scenarios emerge depending on the timing of the classical information and a query. In a pre-measurement scenario, both the query and Alice's answer are exchanged prior to measurement. In this case, knowing already which subset contains the symbol, Bob performs the standard minimum-error discrimination for the subensemble corresponding to the subset $X_\ell$, where  $\ell \in \{ 0,1\}$ is the classical information he receives. The success probability  in the pre-measurement scenario, denoted by $\ppre$, is obtained by maximizing over all $2^{n-1}$ binary partitions $\parti$ and the  measurements $\M$ and $\N$ for each subensemble $\mathcal{E}_0$ and $\mathcal{E}_1$, 
\bea \label{eq:ppre}
\ppre=\max_{ \parti, \M,\N} \sum_{x\in X_0} q_x \tr{\varrho_x \M_x}+\sum_{x\in X_1} q_x \tr{\varrho_x \N_x} 
\eea
where $\sum_{x\in X_0}\M_x=\sum_{x\in X_1}\N_x=\id$.   

Previous works \cite{BaWeWi08,PhysRevA.99.020102,heinosaari2023anticipative,heinosaari2025metainformation} have studied a post-measurement scenario in which the query is fixed before the measurement, but Alice's answer only arrives after it. Bob must therefore decide on his query without knowing what outcome he will get, and we call this an outcome-independent query. Since he knows the classical information is coming later, Bob performs an anticipative measurement $\M=\{\M_z\}$ that maximizes his guessing probability with classical post-processing. Once he has both the measurement outcome $z$ and the classical information $\ell$, he simply guesses the most likely symbol, so the post-processing is $f(\ell,z)=\arg \max_{x} \Pr(x|z,\ell)$. The success probability with the outcome-independent post-measurement query is 
\begin{equation} \label{eq:ppost}
    \begin{split}
        \ppost=&\max_{\M, \parti} \sum_{\ell=0,1} \sum_{x \in X_\ell } q_x \tr{\varrho_x \sum_{z:  f(\ell, z)=x} \M_z}\\
        =&\max_{\M, \parti} \sum_z    \sum_{\ell=0,1}   q_{f(\ell,z)} \tr{\varrho_{f(\ell,z)} \M_z}
    \end{split}
\end{equation}
where the same partition $\parti$ is used for every outcome $z$. It was demonstrated in \cite{PhysRevA.99.020102} that one can even have $\ppost=1$, but this requires that the states belonging to the same subsets in the partitioning are mutually orthogonal. Other examples with $\ppre=\ppost$ were presented in \cite{heinosaari2025metainformation}.
We remark that $\ppre\geq \ppost$ since pre-measurement information can be kept and used after the measurement. 

We now introduce a third, previously unexplored scenario. The query is again posed after the measurement, but, unlike in Eq. \eqref{eq:ppost}, it may depend on the outcome $z$ obtained. We call this an \textit{adaptive} post-measurement query, and denote its success probability by $\pposta$,
\bea \label{eq:pposta}
\pposta=\max_{\M} \sum_z \max_{\{\parti_z\}}   \sum_{\ell=0,1}   q_{f(\ell,z)} \tr{\varrho_{f(\ell,z)} \M_z}
\eea
where the partitions $\{\parti_z\}$ now depend on $z$. 
As we will see, this small tweak makes a big difference.

We illustrate the above scenarios with a simple example where non-orthogonal pure states can be perfectly discriminated only when the query is made after the measurement. Consider a symbol set $X=\{1,2,3\}$ with uniform prior probability, $q_x=\frac{1}{3}$, and the qubit states being
\bea \label{eq:trine}
\psi_x=\frac{1}{\sqrt{2}} \left (
\ket{0}+e^{2 \pi i x/ 3}\ket{1} 
\right )
, \quad x=1,2,3\, .
\eea
These states are also known as the trine states. Since none of the pairs is orthogonal, 
no fixed partition can be used to perfectly discriminate the states, regardless of whether the query is asked before or after the measurement. As all pairwise overlaps are identical, any partition is equally good. For instance, the partition $X_0=\{1\}$ and $X_1=\{2,3\}$ gives the success probability $\ppre=\ppost = \frac{4+\sqrt{3}}{6} \simeq 0.96$ in Eqs. \eqref{eq:ppre} and \eqref{eq:ppost}.  

However, if Bob sends the query after he has obtained a measurement outcome, he can discriminate the states perfectly. The strategy he chooses is known as \textit{antidiscrimination} \cite{CaFuSc02, PuBaRu12,HeKe18}, with the measurement defined as $\M_x=\frac{2}{3}\psi^\perp_x \psi^{\perp *}_x $. An outcome $x$ rules out the state $\psi_x$ with certainty, leaving Bob to identify which of the other two states was sent. 
By separating them into different subsets and asking the corresponding question, he obtains $\pposta=1$.

\textit{The optimal measurement with an adaptive query.  }
In the setting of adaptive post-measurement query, Bob must choose a particular kind of measurement for the task to succeed, leading to the maximal guessing probability $\pposta$.
It is therefore critical to have a recipe for finding the optimal measurement.
Our first main result is that $\pposta$ can be found by considering standard minimum-error discrimination applied to a different state ensemble, which we call the \textit{auxiliary ensemble}. 
First, we define a set of ordered pairs, 
\bea
X_{\mathrm{aux}}=\{ (x,y) \in X \times X| x < y \}
\eea
which consists of a total of $n(n-1)/2$ elements. 
The auxiliary ensemble collects all pairwise mixtures of the original states,
\bea \label{eq:auxens}
\mathcal{E}_{\mathrm{aux}}=\{q_{(x,y)}, \varrho_{(x,y)}\}_{(x,y) \in X_{\mathrm{aux}}}
\eea
where
\bea 
\varrho_{(x,y)}=\frac{q_{x} \varrho_{x}+q_{y} \varrho_{y}}{q_{x} + q_{y}}, \quad q_{(x,y)}=\frac{q_{x}+q_{y}}{n-1}\, .
\eea
We denote by $P_\mathrm{aux}$ the success probability for discriminating the states in the auxiliary ensemble in the standard sense,
\bea
P_\mathrm{aux}=\max_{\M} \sum_{(x,y) \in X_{\mathrm{aux}}} q_{(x,y)} \tr{\varrho_{(x,y)} \M_{(x,y)}}
\eea
where $\M=\{\M_{(x,y)}\}_{(x,y) \in X_{\mathrm{aux}}}$ is a measurement. The following theorem shows that the optimal adaptive post-measurement strategy is directly related to standard minimum-error discrimination of the auxiliary ensemble. 

\begin{proposition}\label{prop1}
    The success probability of the discrimination task with an adaptive post-measurement query is
    \bea
    \pposta =(n-1) P_{\mathrm{aux}}.
    \eea
    The corresponding optimal measurement, denoted by $\madap$, is the minimum-error measurement for the auxiliary ensemble $\mathcal{E}_{\mathrm{aux}}$ in Eq. \eqref{eq:auxens}.
\end{proposition}
	
See Supplemental Material for the proof. Based on Prop. \ref{prop1}, we can now infer the optimal discrimination strategy. Bob performs a measurement that is optimal for the auxiliary ensemble $\mathcal{E}_{\mathrm{aux}}$. 
He records an outcome $z=(x,y)$ and based on that, chooses a partition $X=X_0 \cup X_1$ such that $x\in X_0$ and $y \in X_1$. He guesses that the initial symbol is $x$ or $y$ based on the classical information $\ell=0$ or $1$, respectively.

\indent \textit{ Conditions for an adaptive advantage. }
 Some observations can be made for an ensemble with $n=3$ deciding when $\pposta>\ppre$. It follows that delayed information is always as good, i.e. $\pposta \geq \ppre$, since any bipartition contains a singleton set, whose optimal measurement is a trivial measurement, which is jointly measurable with an optimal measurement for the other subensemble. Then, receiving classical information after the joint measurement is as effective as receiving it earlier, implying $\ppre=\ppost$ and thus $\pposta \geq \ppre$. In fact, the gap between $\ppre$ and $\ppost$ has been shown to characterize measurement incompatibility \cite{UoKrShYuGu19, CaHeTo19, skrzypczyk2019all}.

Secondly, we note that $\pposta>\ppost$ if and only if $\madap$ is necessarily a three-outcome measurement. The reason is that the optimal measurement with a fixed post-measurement query corresponds to a minimum-error measurement for a smaller subset of the auxiliary ensemble \cite{CaHeTo18}. For instance, suppose $\madap$ is a two-outcome measurement with $\madap_{(1,2)}=0$. This adaptive strategy can be implemented by a post-measurement strategy with a fixed partition $X_0=\{1,2\}, X_1=\{3\}$. Conversely, no fixed query strategy can implement $\madap$ if it is a three-outcome measurement.

Based on these observations, we can characterize the structure of pure qubit states for which $\pposta>\ppre$. 
\begin{proposition} \label{prop2}
    For an ensemble of three pure qubit states with uniform prior probability, 
    $$\pposta > \ppre $$ 
    if and only if their Bloch vectors form an acute triangle. 
\end{proposition}

The proof is given in Supplemental Material, which applies geometric method in \cite{bae2013minimum} to the auxiliary ensemble. We illustrate this result with a three-state ensemble $\psi_{1}=\ket{0}, \psi_{2/3}=\cos(\theta/2)\ket{0}\pm\sin(\theta/2)\ket{1}$ for $0<\theta\leq 2 \pi/3$, which includes the trine states in Eq. \eqref{eq:trine}  as a special case. The optimal pre-measurement strategy is to put the most distinguishable pairs, $\psi_{2}$ and $\psi_{3}$, into the same subset, whose optimal measurement is a $\sigma_x$ measurement. This gives $\ppre=\frac{1}{3}(2+\sin\theta)$.

Let us consider post-measurement strategies. For $0<\theta\leq \pi/2$, the Bloch vectors of $\{\varrho\}_{x=1}^3$ form an obtuse or right triangle, and $\madap$ is a two-outcome measurement with $\madap_{(1,2)}=\ketbra{+}, \madap_{(1,3)}=\ketbra{-}$. As this is the same as the minimum-error measurement for subensemble $\{\psi_2, \psi_3\}$, it gives $\ppre=\pposta$. On the other hand, for $\pi/2 < \theta \leq 2\pi/3$, the Bloch vectors form an acute triangle and $\madap$ is a three-outcome measurement, which yields $\pposta=1$ by state-exclusion strategy as in the case of the trine states in Eq. \eqref{eq:trine}.

\textit{Limitations of adaptive queries. }
For $n=3$, it is straightforward to see $\pposta \geq \ppre$ as the pre-measurement strategy involves a trivial measurement. This argument no longer applies to $n \geq 4$. where both subsets may require nontrivial measurements. Remarkably, this ordering is reversed for an ensemble of four qubit states with uniform prior probability.

\begin{proposition} \label{prop3}
    For any ensemble of four pure qubit states with uniform prior probabilities,
    \bea
    \pposta<\ppre.
    \eea
\end{proposition}
Detailed proof is given in Supplemental Material. 
We demonstrate this result with a symmetric informationally complete ensemble that forms a tetrahedron on the Bloch sphere. The states are given by
\bea \label{eq:sicens}
&&\psi_1=\ket{0}, \psi_2=\frac{1}{\sqrt{3}}\ket{0}+\sqrt{\frac{2}{3}}\ket{1},\\
&&\psi_{3/4}=\frac{1}{\sqrt{3}}\ket{0}+e^{\pm 2\pi i /3}\sqrt{\frac{2}{3}}\ket{1}. \nonumber
\eea
Since the pairwise overlaps are identical, any subsets with equal size are equally good, which yields $\ppre=\frac{1}{2}(1+\sqrt{\frac{2}{3}})$. On the other hand, a direct application of Prop. \ref{prop1} gives $\pposta=\frac{1}{2}(1+\frac{1}{\sqrt{3}})$, showing $\ppre>\pposta$.

Although Prop. \ref{prop2} and Prop. \ref{prop3} give contrasting results, one might expect that there always exists an ensemble with $n \geq 4$ exhibiting $\pposta>\ppre$, if the priors are so biased that it is almost as if discriminating three pure states whose Bloch vectors form an acute triangle. We now show that this is indeed the case.

Consider an ensemble $\mathcal{E}=\{1/3,\varrho_x\}_{x=1}^3$ of three pure qubit states whose Bloch vectors form an acute triangle. By Prop. \ref{prop2}, this ensemble exhibits $\pposta(\mathcal{E})>\ppre(\mathcal{E})$, where we explicitly indicate the dependence of the success probabilities on the ensemble. Let us denote the gap between success probabilities by 
\bea \label{eq:gap}
\Delta :=\pposta(\mathcal{E})-\ppre(\mathcal{E}).
\eea
For any $n \geq 4$, we construct an ensemble $\mathcal{E}_{\epsilon}=\{q_x, \tilde{\varrho}_x\}_{x=1}^n$ with prior probabilities 
\bea
q_x=\begin{cases}
    \frac{1-\epsilon}{3}, x=1,2,3\\
    \frac{\epsilon}{n-3}, x=4,...,n,
\end{cases}
\eea
where $\epsilon >0$ is sufficiently small, $\tilde{\varrho}_x=\varrho_x$ for $x=1,2,3$, and the remaining states
can be chosen arbitrarily. 
Then, we have $\lim_{\epsilon \rightarrow 0}\mathcal{E}_\epsilon =\mathcal{E}$.

Using Eq. \eqref{eq:ppre}, we obtain an upper bound 
\bea \label{eq:prebound}
\ppre(\mathcal{E}_\epsilon) \leq (1-\epsilon) \ppre(\mathcal{E}) +  \epsilon
\eea
On the other hand, we can find the lower bound on $\pposta(\mathcal{E}_\epsilon)$ by applying the optimal measurement of $\pposta(\mathcal{E})$ to $\mathcal{E}_\epsilon$, 
\bea
\pposta(\mathcal{E}_\epsilon) &\geq& (1-\epsilon) \pposta(\mathcal{E})\\
&\geq&(1-\epsilon)\Delta + \ppre(\mathcal{E}_\epsilon)-\epsilon \nonumber,
\eea
where the second inequality comes from the bound in Eq. \eqref{eq:prebound}.
Then, we have $\pposta(\mathcal{E}_\epsilon) > \ppre(\mathcal{E}_\epsilon)$ whenever $\epsilon < \frac{\Delta}{1+\Delta}$. We have thus shown that
\textit{for every $n \geq 3$, there exists a qubit ensemble exhibiting $\pposta>\ppre$. }

We illustrate this result with the tetrahedron ensemble in Eq. \eqref{eq:sicens} with $q_1=\epsilon, q_2=q_3=q_4=(1-\epsilon)/3$. To obtain a sufficient condition on $\epsilon$ for which $\pposta>\ppre$, we consider a subensemble $\mathcal{E}=\{\frac{1}{3},\varrho_x \}_{x=2}^4$, whose Bloch vectors form an acute triangle. The gap in Eq. \eqref{eq:gap} is  $\Delta=(2 \sqrt{2}-\sqrt{6})/9\approx0.042$, which gives
\bea\label{eq:epsilon}
\epsilon < \frac{\Delta}{1+\Delta} \approx 0.040.
\eea
For the same ensemble, in pre-measurement strategy, any bipartition into equal-sized subsets is an optimal partition. Since optimal measurements are projective measurements, they are incompatible measurements. Therefore, we obtain the following ordering,
\bea
\pposta>\ppre>\ppost.
\eea
This is an intriguing observation, as it 
shows that the advantage of adaptive post-measurement queries can occur even when the optimal pre-measurement strategy requires incompatible measurements.

\textit{Perfect discrimination of nonorthogonal qudit states. }
In the following, we establish an upper bound for $\pposta$ that can be derived directly from the auxiliary ensemble, describing a basic limitation in state discrimination assisted by a post-measurement query.

\begin{proposition}\label{prop4}
    Let 
    \bea \label{eq:lambda}
    \lambda := \max_{(x,y) \in X_\mathrm{aux}}q_{(x,y)}  \no{\varrho_{(x,y)}}_\infty 
    \eea
    where $\no{\cdot}_\infty$ denotes the operator norm. Then
    \bea
    \pposta \leq d(n-1)\lambda. \label{eq:boundposta}
    \eea
    The upper bound is attained if and only if there exists a measurement $\M$ such that 
    \bea
    q_{(x,y)} \varrho_{(x,y)} \M_{(x,y)}= \lambda \M_{(x,y)} \label{eq:prop2cond}
    \eea
    for all $(x,y) \in \xaux$. 
\end{proposition}
We note that Prop.~\ref{prop4} is a direct consequence of Prop. 1 in \cite{carmeli2022quantum}, specialized here to the auxiliary ensemble. We now explicitly construct an ensemble, with $n=d+1$ states, that saturates the bound in Eq. \eqref{eq:boundposta} and achieves $\pposta>\ppre$. 

The construction uses the \textit{geometrically uniform} (GU) states \cite{eldar2001quantum},
\bea \label{eq:guensemble}
\psi_x=\frac{1}{\sqrt{d}}\sum_{j=0}^{d-1} e^{2 \pi i j x/n}\ket{j}
\eea
with uniform prior probability $q_x=1/n$ for all $x$ and $n=d+1$.  This ensemble has a cyclic group symmetry; a single unitary $U=\sum_{k=0}^{d-1} e^{2 \pi i k/n} \ketbra{k}$ generates the entire ensemble by repeated application to a single state, i.e. $U \psi_x=\psi_{x+1}$ for all $x$ and $U^n=\id$.  Since no two GU states are orthogonal, no pre-measurement query can discriminate them perfectly, so $\ppre<1$ follows immediately. What makes this ensemble special is its complete symmetry: the condition $n=d+1$ forces every pairwise overlap to share the same magnitude,
\bea
|\psi_x^*\psi_y|=\frac{1}{d},~ \forall x \neq y. \label{eq:guoverlap}
\eea

To compute $\pposta$, we construct the auxiliary ensemble in Eq. \eqref{eq:auxens}, 
whose states are
$$\varrho_{(x,y)}=\frac{1}{2}(\psi_x \psi_x^*+\psi_y \psi_y^*)$$
with uniform prior probability $q_{(x,y)}=\frac{2}{n(n-1)}$ for all $(x,y) \in \xaux$. The symmetry in Eq. \eqref{eq:guoverlap} implies that every auxiliary state $\varrho_{(x,y)}$ has the same operator norm,
$\no{\varrho_{(x,y)}}_\infty=\frac{1}{2}(1+1/d)$, with the associated eigenvector
\bea
\varphi_{(x,y)}=\sqrt{\frac{d}{2(d+1)}} (\psi_x+d \psi_y^*\psi_x \psi_y)
\eea
that satisfies $\varrho_{(x,y)} \varphi_{(x,y)}=\no{\varrho_{(x,y)}}_{\infty} \varphi_{(x,y)}$. Since every pair $(x,y)$ attains the maximum in Eq. \eqref{eq:lambda}, the measurement
\bea
\M_{(x,y)}=\frac{2}{d+1} \varphi_{(x,y)}\varphi_{(x,y)}^*, 
\eea
is the optimal measurement that gives $\pposta=1$.

An intuitive way to understand how this measurement can perfectly discriminate the GU states in Eq. \eqref{eq:guensemble}
is to view it through the lens of antidiscrimination. Note that $\tr{\varrho_{s} \M_{(x,y)}}\propto \delta_{s,x}+\delta_{s,y}$; in other words, the outcome $(x,y)$ excludes all the states except $x$ and $y$. By partitioning $x$ and $y$ into different subsets and guessing them based on classical information from Alice, Bob can achieve $\pposta=1$. 

We remark that our result is qualitatively different from the four-qutrit ensemble of Ref.~\cite{PhysRevA.99.020102}. Perfect discrimination with post-measurement information follows from orthogonality within a fixed partition, resulting in $\ppre=\ppost=1$. In the GU ensemble of Eq. \eqref{eq:guensemble}, by contrast, no pair of states is orthogonal, so adaptivity of the query is essential, enabling post-measurement query to outperform the pre-measurement strategy, $1=\pposta>\ppre$. 

\textit{Maximal number of perfectly distinguishable states. }
In communication scenarios, the maximal number of perfectly distinguishable states is an essential characteristic feature of a quantum system. Without any classical side information, the maximal number of perfectly distinguishable states is obviously $d$. Allowing an outcome-independent post-measurement query raises this limit only slightly \cite{PhysRevA.99.020102}, to 
\bea \label{eq:bound}
n \leq d+1.
\eea
Our GU construction shows that this larger bound is already achievable with $n=d+1$ states that are perfectly discriminable only with adaptive post-measurement query, giving $\pposta=1$ while $\ppre<1$. 
An immediate question is whether the bound in Eq. \eqref{eq:bound} still holds for outcome-dependent queries, or whether we can surpass it.

When $d=2$, it is clear that $\pposta<1$ for four or more states. This is because to achieve $\pposta=1$ every measurement outcome must exclude $n-2$ states, which is not possible for qubits unless $n\leq 3$.  We next show that the same limit persists one dimension higher: for $d=3$, $\pposta=1$ remains possible only for $n \leq 4$. 

\begin{proposition} \label{prop5}
    For $d\leq 3$, $\pposta=1$ is possible only if $n \leq d+1$.
\end{proposition}
A detailed proof is given in the Supplemental Material. Prop. \ref{prop5} may hint that the bound in Eq. \eqref{eq:bound} still holds for outcome-dependent queries in general. However, the following construction of states shows that the bound can be surpassed, demonstrating that an outcome-dependent query can be more powerful than an outcome-independent one. 

The construction relies on the fact that discrimination problems in orthogonal subspaces can be combined. For $d=2m$ with $m$ a positive integer, define
\bea \label{eq:trineextended}
\psi_{3j+k}=\frac{1}{\sqrt{2}} \left( \ket{2j}+e^{2 \pi i k/3} \ket{2j+1} \right)
\eea
for $j=0,...,m-1$ and $k=1,2,3$. These are $m$ trine ensembles as in Eq. \eqref{eq:trine}, one embedded in each of the $m$ orthogonal two-dimensional subspaces spanned by $\{\ket{2j}, \ket{2j+1}\}$. Within each subspace, the trine states can be perfectly antidiscriminated, and coarse-graining these antidiscriminating measurements gives a measurement that perfectly discriminates the whole ensemble, so $\pposta=1$ for $n=\frac{3}{2}d$ states. 

For odd $d=2m+1$, the same idea applies with one modification. We use the trine ensemble in Eq. \eqref{eq:trineextended} for $j=0,...,m-2$, and fill the remaining three-dimensional subspace with an ensemble of four qutrit states that achieves $\pposta=1$, such as the ensemble in Eq. \eqref{eq:guensemble}. 

We have thus shown a systematic construction of ensembles with
\bea
n= \left \lfloor \frac{3d}{2} \right \rfloor
\eea
states satisfying $\pposta=1$. Since no bipartition of these ensembles yields two subensembles that each contain only orthogonal states, $\ppre<1$ follows. For $d>3$, this construction exceeds the bound $n\leq d+1$ that holds for outcome-independent queries, showing that adaptivity extends what is achievable.

\textit{Discussion. }
In conclusion, we have introduced a framework for query-assisted state discrimination and shown that an adaptive post-measurement query can outperform any pre-measurement queries, even enabling perfect discrimination of pairwise nonorthogonal states. Our main result provides a general method for finding the optimal strategy by reducing the problem to standard minimum-error discrimination of an auxiliary ensemble.

Our results challenge the usual intuition that  later classical information is less useful than earlier information. 
The key point is that the query can depend on the measurement outcome, allowing the receiver to choose which classical information to request. We have shown that this advantage depends on the structure of the state ensemble and that adaptive queries can exceed the previously known dimensional bound on the number of perfectly distinguishable states with outcome-independent queries.

It would be interesting to characterize the general structure of ensembles for which $\pposta>\ppre$. We conjecture that $\ppre>\pposta$ for any qubit ensembles of $n \geq 5$ pure states with uniform prior probabilities. Another interesting direction is to determine the maximal number of perfectly distinguishable states using an adaptive query.

Since state discrimination is a basic primitive underlying more complex information-processing tasks \cite{RevModPhys.74.145,scarani2009security}, our framework offers a new perspective on how classical information can be used most effectively in quantum information processing.

\section{Acknowledgement}
The authors acknowledge financial support from the Business Finland project BEQAH. We thank Takayuki Miyadera for valuable discussions. 

\bibliographystyle{apsrev4-2}
\bibliography{reference}

\newpage
\onecolumngrid
\appendix
\section*{Supplemental Material}

\subsection*{I. Proof of Proposition 1}
In this section, we prove that the optimal measurement with the adaptive post-measurement query is the minimum-error measurement for the auxiliary ensemble. We consider an ensemble of states $d$ dimensional quantum states
\bea \label{app:ens}
\mathcal{E}=\{q_x, \varrho_x\}_{x=1}^n
\eea
where $q_x$ denotes prior probability. We denote $\M=\{\M_z\}_{z=1}^m$ an $m$ outcome measurement, $\parti_z$ an outcome depenent partition, and $f(\ell,z)$ a classical post-processing map that outputs an element in the symbol set $X$. In particular, we only consider disjoint non-empty subsets and deterministic post-processing. Let us first consider a fixed partition
\bea
\parti: X=X_0 \cup X_1
\eea
For given $\M,\parti,$ and $f$, the success probability with the post-measurmeent information \cite{heinosaari2025metainformation} is written as
\bea
P=\sum_{\ell} \sum_{x \in X_\ell} q_x \tr{\varrho_x \sum_{z:f(\ell,x)=z} \M_z}= \sum_{z=1}^m    \sum_{\ell}   q_{f(\ell,z)} \tr{\varrho_{f(\ell,z)} \M_z}
\eea
We remark that the success probability takes the same form even when the partition is dependent on the measurement outcome. With an adaptive post-measurement query, the success probability is given by
\bea
\pposta=\max_{\M} \sum_z \max_{\{\parti_z\}, f} \sum_\ell q_{f(\ell,z)} \tr{\varrho_{f(\ell,z)} \M_z}
\eea
where $\{\parti_z\}$ are the partitions that depend on the measurement outcomes. 

We now define an auxiliary symbol set $
\xaux=\{(x,y) \in X \times X | x < y\}$, which consists of an ordered pair in $X\times X$. The auxiliary ensemble is defined as
\bea
\mathcal{E}_{\mathrm{aux}}=\{q_{(x,y)}, \varrho_{(x,y)}\}_{(x,y) \in \xaux}
\eea
where
\bea
\varrho_{(x,y)}=\frac{q_x \varrho_x + q_y \varrho_y}{q_x + q_y}, ~q_{(x,y)}=\frac{q_x + q_y}{n-1}
\eea
is a pairwise mixture of the original states and $q_{(x,y)}$
denotes its prior probability. Note that the factor $1/(n-1)$ is multiplied to meet the normalization constraint. Define a map
\bea
\pi(z)=\arg\max_{(x,y) \in \xaux} \tr{q_{(x,y)} \varrho_{(x,y)} \M_z}
\eea
that outputs an element in $\xaux$. Then, the upper bound for the success probability is derived as follows,
\bea
P \leq \sum_z q_{\pi(z)} \tr{\varrho_{\pi(z)} \M_z}=(n-1)\sum_{(x,y) \in \xaux} q_{(x,y)} \tr{\varrho_{(x,y)} \M_{(x,y)}}
\eea
where $\M_{(x,y)}=\sum_{z \in \pi^{-1}(x,y)}\M_z$
forms a POVM, i.e. $\M_{(x,y)} \geq 0$ and $\sum_{(x,y)}\M_{(x,y)}=\id$. Then, 
\bea
\pposta=\max_{\M, \{\parti_z\},f} P \leq (n-1) \max_{\N} \sum_{(x,y) \in \xaux} q_{(x,y)} \tr{\varrho_{(x,y)} \N_{(x,y)}}=(n-1) \paux. 
\eea
where $\paux$ denots the success probability of the auxiliary ensemble. 

In fact, this upper bound can always be attained by making the minimum-error measurement $\M=\{\M_{(x,y)}\}$ for the $\mathcal{E}_{\mathrm{aux}}$ and separating $x$ and $y$ into different subsets, and picking them based on the classical information from Alice. Therefore, the optimal strategy with the adaptive post-measurement query is essentially the minimum-error discrimination task for the auxiliary ensemble. 

\subsection{II. Geometric method to minimum-error discrimination}
In this section, we review the geometric approach \cite{bae2013minimum,Bae13} to minimum-error discrimination for qubit ensembles with equal prior probabilities, and apply this to prove Propositions \ref{prop2} and \ref{prop3}. Consider an ensemble of qubit states $\{\varrho_x\}_{x=1}^n$ with uniform prior probabilities, $q_x=1/n$ for all $x$. Let $\vec{v}_x$ denote the Bloch vector of $\varrho_x$, and let $\mathsf{P}$ be a polytope formed by $\{\vec{v}_x\}_{x=1}^n$. Then, the optimal success probability for the ensemble is given by
\bea
P=\max_{\M} \sum_{x=1}^n \frac{1}{n}\tr{\varrho_x \M_x}=\frac{1}{n}(1+R) \label{eq:optsuc}
\eea
where $R$ is the radius of the smallest enclosing sphere of $\mathsf{P}$. The optimal measurement is 

\subsubsection{A. Proof of Proposition 2}
In this section, we prove Proposition \ref{prop2} in the main manuscript. 
\\

\textbf{Proposition 2. }
  \textit{  For an ensemble of three pure qubit states with uniform prior probability, 
    $$\pposta > \ppre $$ 
    if and only if their Bloch vectors form an acute triangle. }

\begin{proof}
    In Ref. \cite{Bae13},  it was shown that, for three qubit states with uniform prior probabilities, every optimal measurement has three nonzero outcomes if and only if their Bloch vectors form an acute triangle.
    The auxiliary states have the Bloch vectors $\vec{v}_{(x,y)}=\frac{1}{2}(\vec{v}_x + \vec{v}_y)$. Let us denote by $\no{\vec{v}_{(1,2)}-\vec{v}_{(1,3)}}$ the longest side of the triangle formed by the Bloch vectors of the auxiliary ensemble. The condition for $\madap$ to have three outcomes is then $\no{\vec{v}_{(1,2)}-\vec{v}_{(1,3)}}^2<\no{\vec{v}_{(1,2)}-\vec{v}_{(2,3)}}^2+\no{\vec{v}_{(1,3)}-\vec{v}_{(2,3)}}^2$, which is equivalent to $\no{\vec{v}_1-\vec{v}_3}^2<\no{\vec{v}_2-\vec{v}_3}^2+\no{\vec{v}_1-\vec{v}_2}^2$. Therefore, the Bloch vectors $\{\vec{v}_x\}_{x=1}^3$ must form an acute triangle.  
\end{proof}

\subsubsection{B. Proof of Proposition 3}
Here we prove Proposition \ref{prop3} in the main manuscript.
\\

\textbf{Proposition 3.}
    \textit{For any ensemble of four pure qubit states with uniform prior probabilities,
    \bea
    \pposta<\ppre.
    \eea}
\begin{proof}
Let us consider an ensemble of four pure qubits $\{\varrho_x\}_{x=1}^4$ with uniform prior probabilities, $q_x=1/4$ for all $x$. There are total $6$ auxiliary states. Let $\vec{v}$ denote the Bloch vector of the ensemble state, $\varrho=\frac{1}{4}\sum_x \varrho_x$. We observe that the Bloch vectors of the auxiliary states, $\vec{v}_{(x,y)}$, are located in antipodal with respect to $\vec{v}$. This is shown by 
\bea
\vec{v}_{(1,2)}-\vec{v}=\frac{1}{4}(\vec{v}_1+\vec{v}_2-\vec{v}_3-\vec{v}_4)=\vec{v}-\vec{v}_{(3,4)}
\eea
which holds for any pairs instead of $(1,2)$. This observation renders $\pposta$ in a simple form using Eq. \eqref{eq:optsuc},
\bea
\pposta=\frac{1}{2}(1+\max_{x,y} \no{\vec{v}_{(x,y)}-\vec{v}})
\eea
Without loss of generality, let us choose the pair $(1,2)$ the maximizer. Then,
\bea
\pposta&=&\frac{1}{2}+\frac{1}{8}\no{\vec{v}_1+\vec{v}_2-\vec{v}_3-\vec{v}_4}\\
&\leq & \frac{1}{2}+\frac{1}{8}(\no{\vec{v}_1-\vec{v}_3}+\no{\vec{v}_2-\vec{v}_4})\\
&=& \frac{1}{2}+\frac{\sqrt{2}}{8}( \sqrt{1-\vec{v}_1\cdot\vec{v}_3}+\sqrt{1-\vec{v}_2\cdot\vec{v}_4})\\
&=&\frac{1}{2}+\frac{1}{4}(\sqrt{1-|\psi_1^*\psi_3|^2}+\sqrt{1-|\psi_2^*\psi_4}|^2)\\
&\leq& \ppre,
\eea
where the last inequality comes from that the success probability corresponds to a pre-measurement strategy with a particular partition $X_0=\{1,3\}, X_1=\{2,4\}$. 

Now let us prove the strict inequality. Suppose $\pposta=\ppre$. This implies that the optimal partition for $\ppre$ is a partition of equal sizes, for instance $X_1=\{1,3\}, X_2=\{2,4\}$, and the inequality becomes equality
\bea
\no{\vec{v}_1+\vec{v}_2-\vec{v}_3-\vec{v}_4}=\no{\vec{v}_1-\vec{v}_3}+\no{\vec{v}_2-\vec{v}_4}
\eea
This implies that $\vec{v}_1-\vec{v}_3$ and $\vec{v}_2-\vec{v}_4$ are parallel. We can write
\bea
\vec{v}_1-\vec{v}_3 = \alpha \vec{e}, ~\vec{v}_2-\vec{v}_4 = \beta \vec{e}
\eea
with $\alpha, \beta >0$ and $\vec{e}$ is a unit vector. We can paremetrize the vectors with the normal vectors $\vec{n}, \vec{m}$ such that $\vec{n}\cdot \vec{e}=\vec{m}\cdot\vec{e}=0$,
\bea
\vec{v}_1=\vec{n}+\frac{\alpha}{2}\vec{e}, ~\vec{v}_2=\vec{m}+\frac{\beta}{2}\vec{e},
~\vec{v}_3=\vec{n}-\frac{\alpha}{2}\vec{e}, 
~\vec{v}_4=\vec{m}-\frac{\beta}{2}\vec{e}.
\eea
Then,
\bea
\no{\vec{v}_1-\vec{v}_4}+\no{\vec{v}_2-\vec{v}_3}=\sqrt{(\alpha+\beta)^2+4\no{\vec{n}-\vec{m}}^2}>\alpha+\beta=\no{\vec{v}_1-\vec{v}_3}+\no{\vec{v}_2-\vec{v}_4}
\eea
This leads to the contradiction that $\ppre>\pposta$ since
\bea
\ppre\geq \frac{1}{2}+\frac{1}{4}(\sqrt{1-|\psi_1^*\psi_4|^2}+\sqrt{1-|\psi_2^*\psi_3|^2})=\frac{1}{2}+\frac{1}{8}(\no{\vec{v}_1-\vec{v}_4}+\no{\vec{v}_2-\vec{v}_3})>\pposta
\eea
\end{proof}

\subsection*{ III. Proof of Proposition 5} 
Here we prove the limit of the number of perfectly distinguishable states using an adaptive post-measurement query for $d\leq 3$. Proposition \ref{prop3} states that $\pposta=1$ is possible only if $n\leq d+1$ for $d \leq 3$. Note that its converse is trivially true since an ensemble consisting of $d$ orthogonal states and any one additional state achieves $\pposta=1$. It is also clear that when $d=2$, the maximal number of perfectly distinguishable states is $3$. This follows from the fact that each measurement outcome must exclude at least $n-2$ states, which is not possible for qubit states when $n \geq 4$. 

Consider the ensemble in $\eqref{app:ens}$. A useful notion related to state discrimination is \textit{communication matrix} \cite{FrWe15}, $C=[C_{x,z}]$, whose elements denote the conditional probability of measurement outcome given state preparation, $C_{x,z}=\tr{\varrho_x \M_z}$. To achieve $\pposta=1$, a measurement must implement a communication matrix, in which each column must contain at least $n-2$ zeros. 

Before we begin the actual proof, let us establish some useful lemmas. From now on, we consider communication matrices with rank-one effects.
\begin{lemma} \label{lemma1}
    Consider three states whose support lie in a two-dimensional subspace, denoted by $\varrho_1, \varrho_2$ and $\varrho_3$. If there exists a vector $v$ which is orthogonal to two of the states, for instance $\varrho_1 v = \varrho_2 v =0$, then the vector must be orthogonal to the remaining state, $\varrho_3 v=0$.  
\end{lemma}
\begin{proof}
    Any two states chosen from $\{\varrho_1, \varrho_2, \varrho_3\}$ span the two-dimensional subspace (otherwise they must be identical pure states). Therefore, if $\varrho_1v=\varrho_2v=0$, then $v$ must be orthogonal to the subspace, and therefore $\varrho_3v=0$. 
\end{proof}

\begin{lemma} \label{lemma2}
    Consider a communication matrix $C$ with $d=3$. If there exist three zero elements in a single column, for instance $C_{1,j}=C_{2,j}=C_{3,j}=0$ for some $j$, then for all $z$, it is impossible that only two of $\{C_{1,z},C_{2,z},C_{3,z}\}$ are zeros. 
\end{lemma}
\begin{proof}
    This lemma is a direct consequence of Lemma \ref{lemma1}. The existence of $j$ in which $C_{1,j}=C_{2,j}=C_{3,j}=0$ holds implies that the three qutrits $\{\varrho_1,\varrho_2, \varrho_3\}$ lie in two-dimensional subspace. By Lemma \ref{lemma1}, there cannot exist a column $z$ where only two of $\{C_{1,z},C_{2,z},C_{3,z}\}$ are zeros.
\end{proof}

We are now ready to prove Proposition \ref{prop5}. 
\\

\textbf{Proposition 5. }{ \it For $d\leq 3$, $\pposta=1$ is possible only if $n \leq d+1$. }
\begin{proof}
    We will prove the claim for $d=3$ by contradiction that a communication matrix satisfying $\pposta=1$ cannot exist. Suppose that $\pposta=1$ can be achieved with $n \geq 5$. Each column of the communication matrix must contain at least $n-2$ zeros. We first note that each effect $\M_z$ must be rank-one because if it is rank-two then all states that satisfy $\tr{\varrho_x \M_z}=0$ must be identical. Without loss of generality, we can write the first five columns and five rows of the communication matrix as 
    \bea
    C=
    \begin{pmatrix}
        \star  & ~ & ~ & ~ & ~ & \ldots\\
        ~ & \star &  ~& ~ & \\
        &  & \star&  & \\
        &   &   &\star &  \\
        &   &   &   & \star\\
        \vdots
    \end{pmatrix}
    \eea
    where each $\star$ denotes a positive value, not necessarily the same in different entries. The empty spaces are considered as non-negative numbers. Since each column can have at most two non-zero elements, we write the first column as follows, 
    \bea
    \begin{pmatrix}
        \star  &  &  & &  & \ldots\\
        \diamond& \star & & &\\
        0 & & \star& &\\
        0& &  & \star &\\
        0& &  & & \star\\
        \vdots
    \end{pmatrix}
    \eea
    where $\diamond \geq 0$. According to Lemma \ref{lemma2}, the elements $C_{4,3}$ and $C_{5,3}$ cannot both be zeros. 
    As this argument applies to the fourth and the fifth columns as well and each column can have at most two non-zero elements, we can fill the communication matrix as 
    \bea
    \begin{pmatrix}
        \star  &  & 0 &0 & 0 & \ldots\\
        \diamond& \star & 0&0 &0\\
        0 & & \star& &\\
        0& & \alpha & \star &\\
        0& & \beta & & \star\\
        \vdots
    \end{pmatrix}
    \eea
    where exactly only one of $\alpha$ or $\beta$ is zero. However, the choice of $\alpha>0$ and $\beta=0$ violates Lemma \ref{lemma1} since $\{\varrho_1, \varrho_2, \varrho_4\}$ are in a two-dimensional subspace but $C_{1,4}=C_{2,4}=0$ and $C_{4,4}>0$. Likewise, the choice of $\alpha=0$ and $\beta>0$ violates Lemma \ref{lemma1} by analogy with the fifth column. Therefore, no communication matrix satisfying $\pposta=1$ can exist for $d=3$ and $n \geq 5$. 
\end{proof}

\end{document}